\documentclass[prd,floats,floatfix,showpacs,nofootinbib]{revtex4}
\usepackage{graphicx}
\usepackage{dcolumn}
\usepackage{bm}
\usepackage{graphics}
\usepackage{slashed}
\usepackage{amssymb}
\usepackage{natbib}
\usepackage{amsmath}
\usepackage{amsthm}
\usepackage{url}
\usepackage{xcolor}
\usepackage[normalem]{ulem}
\usepackage{hyperref}
\newcommand{\fracc}[2]{\frac{\textstyle{#1}}{\textstyle{#2}}}

\newcommand\redout{\bgroup\markoverwith{\textcolor{red}{\rule[.5ex]{2pt}{0.6pt}}}\ULon}
\newtheorem{cor}{Corollary}
\begin{document}

\title{On the disformal invariance of the Dirac equation}

\author{Eduardo Bittencourt}\email{Eduardo.Bittencourt@icranet.org}
\author{Iarley P.\ Lobo}\email{Iarley.PereiraLobo@icranet.org}
\author{Gabriel G.\ Carvalho}\email{Gabriel.Carvalho@icranet.org}

\affiliation{CAPES Foundation, Ministry of Education of Brazil, Bras\'ilia, Brazil and\\
Sapienza Universit\`a di Roma - Dipartimento di Fisica\\
P.le Aldo Moro 5 - 00185 Rome - Italy}
\date{\today}

\begin{abstract}
We analyze the invariance of the Dirac equation under disformal transformations depending on the propagating spinor field acting on the metric tensor. Using the Weyl-Cartan formalism, we construct a large class of disformal maps between different metric tensors, respecting the order of differentiability of the Dirac operator and satisfying the Clifford algebra in both metrics. We split the analysis in some cases according to the spinor mass and the norm of the Dirac current, exhibiting sufficient conditions to find classes of solutions which keep the Dirac operator invariant under the action of the disformal group.


\end{abstract}

\maketitle

\section{Introduction}
In the nineties the disformal transformations appeared with some notoriety in the literature through Bekenstein's works \cite{beken1,beken2}, where the possibility of adding more than one Riemannian geometry to a geometrical theory of gravity is revisited and discussed in the realm of the so called Finsler geometries. Since then, investigations of new kinds of symmetry--beyond the well-known external (rotational, translations, boosts etc) and internal (gauge) ones--under which a given dynamical equation could be invariant have increased. As a consequence of this, the disformal transformations have been used with the aim of explaining some of the current open problems in physics, for instance, MOND \cite{beken_mond,milgrom}, modified dispersion relations in quantum gravity phenomenology \cite{amelino}, bimetric theories of gravity \cite{clifton}, scalar-tensor theories \cite{dario}, disformal inflation \cite{nemanja}, chiral symmetry breaking \cite{bitt_nov_faci}, anomalous magnetic moment for neutrinos \cite{nov_bit}, analogue models of gravity \cite{nov_bit_gordon,nov_bit_drag} and others.

The disformal transformations are usually defined in a scenario with a manifold ${\cal M}$ endowed with metrics $g_{\mu\nu}$ and $\widehat g_{\mu\nu}$. Assuming a field $\Phi$ (scalar, vector or spinor) satisfying a given dynamics defined in terms of $g_{\mu\nu}$ on ${\cal M}$, we can induce a dynamical equation for $\Phi$ which takes into account only $\widehat{g}_{\mu\nu}$, or vice versa, if both metric tensors are related. In general, the disformal map is assumed to be $({\cal M},g,\Phi)\mapsto ({\cal M},\widehat{g},\Phi)$ with
\begin{equation}
\label{disf_met_intro}
\widehat g_{\mu\nu}=\alpha(\Phi,\nabla\Phi)g_{\mu\nu} + \Sigma_{\mu\nu}(\Phi,\nabla\Phi),
\end{equation}
where the arbitrary function $\alpha$ is positive definite for any point on ${\cal M}$ and the rank-2 tensor $\Sigma_{\mu\nu}$ depends on $\Phi$ and possibly on its covariant derivatives defined with the pseudo-Riemannian connection associated to the metric $g_{\mu\nu}$. That is, when the dynamics of $\Phi$ defined in terms of $g_{\mu\nu}$ on ${\cal M}$ is rewritten in terms of $\widehat g_{\mu\nu}$, we obtain its propagation with respect to this metric and, therefore, the two dynamics for $\Phi$ are said to be \textit{disformally} equivalent \cite{novetgoul,goulart}.

In this work, neither the disformal metric $\widehat g_{\mu\nu}$ nor the target metric $g_{\mu\nu}$ have any gravitational character (these names will become clear afterwards) and, thereupon, we do not need to impose any dynamics to them. They are seen as the substrata given {\it a priori} where the external field is allowed to propagate. This is crucial for the statement that we are dealing with the same propagating field the whole time and that these two different representations of its dynamics are indeed equivalent. In other words, there is a degeneracy on the choice of the space-time metric. Besides, in this context, we can affirm that this kind of correspondence is helpful in the search of solutions for special classes of PDEs, since a field configuration satisfying its propagating equation in a given metric also verifies the dynamical equation of its disformal equivalent representation.

According to the lines presented in the scalar and electromagnetic cases \cite{erico12,erico13}, we focus here on the disformal transformations applied to spinor fields or, more precisely, we establish under which circumstances one can obtain the disformal invariance of the Dirac equation defined in an arbitrary background metric. When the Dirac current is time-like, we will show that these transformations can be easily implemented to a class of spinors if we use a generalization of Inomata's condition (cf.\ for instance \cite{Inomata}), whose covariant derivative of the spinor is written as a linear combination of the elements of the Clifford algebra. In the case of light-like currents, the conditions for the disformal invariance of the Dirac equation lie only on the properties of the current itself and the tetrad frame used to provide the disformal map.

This paper is organized as follows. In Sec.\ [\ref{math-prelim}], we provide an intrinsic meaning to the disformal maps, i.e., without referring to any coordinate system explicitly. In Sec.\ [\ref{disf_trans}], we introduce the disformal transformations via the Weyl-Cartan formalism in order to work properly with spinor fields in curved space-times. In Sec.\ [\ref{cons_curr_sec}], we analyze the conservation laws of the vector and axial currents in both geometries to guarantee also the physical equivalence of the disformal map. In Sec\ [\ref{DB}], we then define the Dirac equation in the disformal metric and, using the formulas introduced in Sec.\ [\ref{disf_trans}], we rewrite it in the target metric obtaining a nonlinear equation for the spinor field and, in Secs.\ [\ref{spec-clas}] and [\ref{null-case}], we analyze the possible cases exhibiting a class of spinors for which the disformal invariance of the Dirac equation holds true. For completeness, in Sec.\ [\ref{disf_group}] we prove that the disformal maps associated with spinor fields satisfy an Abelian group structure.

\section{Intrinsic definition of the disformal maps}\label{math-prelim}
Traditionally, the disformal transformation is explicitly constructed with the components of the tensors involved in the map (see Eq.\ \ref{disf_met_intro}) and, in principle, this could lead to ambiguities since the disformal maps work easier with the contra-variant components of the metric rather than the formally defined covariant rank 2 tensor. Thus, in this section, we shall define the disformal transformation in an intrinsic way to show that in fact one can choose any representation (covariant or contra-variant) for the disformal metric without loss of generality. Following the lines we shall use throughout the text, we set:
\newtheorem{defi}{Definition}
\begin{defi}
A (pseudo-Riemannian) metric tensor field $g$, in a differentiable manifold ${\cal M}$, is a smooth, symmetric, bilinear, non-degenerate map which assigns a real function to pairs of vector fields of the tangent bundle associated with ${\cal M}$, that is
\begin{equation*}
\begin{array}{lrcl}
g\,:&\Gamma(T{\cal M})\bigotimes \Gamma(T{\cal M})&\rightarrow& {\cal F}({\cal M}),\\[1ex]
&(X,Y)&\mapsto& g(X,Y).
\end{array}
\end{equation*}
\end{defi}
By non-degenerate we mean that the determinant of the matrix formed by the metric components is nonzero at each point of the manifold, namely
$\det [g (\partial_{\mu} , \partial_{\nu} )] |_{p} \neq 0,$ where $\partial_{\mu}$ represent the elements of the vector basis in the tangent space at each $p\in {\cal M}$ and from which we define the metric components $g_{\mu\nu}\doteq g(\partial_{\mu},\partial_{\nu})$. We assume that the metric tensor $g$ has the Lorentzian signature $(+,-,-,-)$.

\newtheorem{prop}{Proposition}\label{prop1}
\begin{prop}
Let ${\cal M}$ be a space-time with metric tensor $g$. Consider also the smooth vector fields $J$ and $I$, such that $g(J,J)=-g(I,I)\doteq J^2$, $g(J,I)=0$ and the regular scalar fields $\alpha, \beta, \gamma$ and $\delta$ depending functionally on $J$ and $I$, such that $\alpha$ is positive definite and $\Upsilon^2=(\alpha+\beta)(\alpha+\gamma) + \delta^2$ is never null. Then, the map
\begin{equation}
\begin{array}{lrcl}
\widehat g\,:&\Gamma(T{\cal M})\bigotimes \Gamma(T{\cal M})&\rightarrow& {\cal F}({\cal M}),\\[1ex]
&(X,Y)&\mapsto& \widehat{g}(X,Y).
\end{array}
\end{equation}
with
\begin{eqnarray}
\widehat{g}(X,Y)\doteq \alpha\, g(X,Y)+\frac{\beta}{J^2}\,g(J,X)\,g(J,Y)-\frac{\gamma}{J^2}\,g(I,X)\,g(I,Y)+\frac{\delta}{J^2}\,[g(J,X)\,g(I,Y)+g(I,X)\,g(J,Y)]
\end{eqnarray}
or, equivalently,
\begin{equation}
\widehat{g}(\ast,\cdot)=\alpha\, g(\ast,\cdot)+\frac{\beta}{J^2}\,g(J,\ast)\otimes g(J,\cdot)-\frac{\gamma}{J^2}\,g(I,\ast)\otimes g(I,\cdot)+\frac{\delta}{J^2}\,[g(J,\ast)\otimes g(I,\cdot)+g(I,\ast)\otimes g(J,\cdot)]
\end{equation}
is also a pseudo-Riemannian metric for ${\cal M}$. Furthermore, if we require $\alpha>0$, $\alpha+\beta>0$ and $\alpha+\gamma>0$, implying that $\Upsilon^2>0$, the causal character is maintained. This is the {\rm disformal transformation} of the metric $g$, and in this case $\widehat{g}$ is called the {\rm disformal metric} associated with the {\rm target metric} $g$.
\end{prop}
\begin{proof}
It is straightforward from Definition $1$ by using Cayley-Hamilton's formula. From this, we derive that the determinant of the disformal metric is non-degenerate only if $\Upsilon^2\neq 0$. The requirements for keeping the causal structure are obtained if we choose an orthonormal basis $\cal{B}$ of the tangent space constructed in terms of $J$ and $I$, completing this basis, and evaluating $\widehat{g}$ in the elements of $\cal{B}$.
\end{proof}
\par

The eigenvalue problem associated with the quantity $\widehat{g}^{\mu}{}_{\ \nu}\equiv\widehat{g}^{\mu\alpha}{g}_{\alpha\nu}$ provides the following nontrivial eigenvalues
\begin{gather}
\lambda_{\pm}= \alpha + \frac{ \gamma+\beta \pm \sqrt{(\beta-\gamma)^2 -4\delta^2}}{2},
\end{gather}
with eigenvectors $v_{\pm}$ given by
\begin{gather}
v_{\pm}=J+\frac{\beta-\lambda_{\pm}}{\delta}\,I.
\end{gather}
The other two eigenvalues are degenerated and equal to $\alpha$, with eigenvectors lying on the orthogonal complement of $J$ and $I$.

Moreover, from Levi-Civita's theorem, a given metric tensor on the manifold possesses a unique affine connection, called Riemannian connection ($\nabla$), satisfying the requirements of symmetry and compatibility with the given metric (i.e., $\nabla g = 0$). Therefore, for the metric tensors $g$ and $\widehat{g}$ on ${\cal M}$ we can relate the Riemannian connections $\nabla$ and $\widehat{\nabla}$, respectively and unequivocally.

From the proposition above, one sees that the attribution of space-time components to the metrics associated with the disformal map does not lead to contradictions because the map itself has an intrinsic definition. For this reason, one can use in practice either the contra-variant or the covariant components of both tensors, according to convenience.

\section{Construction of the disformal transformation for spinors}\label{disf_trans}
In this section we develop the mathematical tools we need to deal with spinor fields in a curved space-time. As ensured by Proposition $1$, we shall use space-time components to favour the calculations and to present our results in a standard language. We mainly concentrate on the Clifford algebra that the Dirac matrices must satisfy in the disformal and target metrics and on the construction of the disformal map between these metrics through the Weyl-Cartan formalism.

With the Dirac matrices $\gamma^{\mu}$ and an arbitrary Dirac spinor field (bispinor) $\Psi$ defined in a space-time metric $g_{\mu\nu}$, one can construct two Hermitian scalars $A\equiv\bar\Psi \Psi$ and $B\equiv i \bar\Psi \gamma_5 \Psi$, where $\bar \Psi \equiv \Psi^{\dag} \gamma^0$ and $\gamma_5 \equiv \frac{i}{4!}\,\eta_{\alpha\beta\mu\nu} \gamma^\alpha \gamma^\beta \gamma^\mu \gamma^\nu$ with $\eta_{\alpha\beta\mu\nu}$ corresponding to the Levi-Civita tensor. Throughout the text the Dirac representation for $\gamma^{\mu}$ is assumed. We also define two space-time vectors depending algebraically on $\Psi$ which are the Dirac current $J^{\mu} \equiv \bar\Psi_{} \gamma^{\mu} \Psi_{}$ and the axial current $I^{\mu} \equiv \bar\Psi_{} \gamma^{\mu} \gamma_5 \Psi_{}$. Using Pauli-Kofink's identity \cite{spin_curv}
\begin{equation}
\label{pkofink}
(\bar\Psi Q\gamma_{\lambda}\Psi)\gamma^{\lambda}\Psi=(\bar\Psi Q\Psi)\Psi - (\bar\Psi Q\gamma_5\Psi)\gamma_5\Psi,
\end{equation}
where $Q$ is an arbitrary element of the Clifford algebra, one can see that the currents are related to the scalars through $J^{\mu}J_{\mu}=-I^{\mu}I_{\mu}=A^2 + B^2$ and $J_{\mu}I^{\mu}=0$.

According to the spin-2 field theory formulation \cite{feynman}, any space-time metric $\widehat{g}_{\mu\nu}$ can be always split into two parts
\begin{equation}
\label{met}
\widehat{g}_{\mu\nu} = g_{\mu\nu} + \Sigma_{\mu\nu},
\end{equation}
with a background geometry $g_{\mu\nu}$ and a rank two tensor field $\Sigma_{\mu\nu}$ responsible for the spin-2 particle description. Under this form, the inverse metric $\widehat{g}^{\mu\nu}$ is given by an infinite series, in general. However, in the case of disformal metrics expressed by Eq.\ (\ref{met}), their inverse admits the same binomial form if we choose $\Sigma_{\mu\nu}$ such that the condition $\Sigma^{\mu\nu} \, \Sigma_{\nu\lambda} = p \, \delta^{\mu}_{\lambda} + q \, \Sigma^{\mu}{}_{\lambda}$ holds, where $p$ and $q$ are arbitrary functions of the coordinates.

In this way, we shall proceed with the disformal transformation using directly the components of the space-time geometry, as we said before. Thereby, once the Dirac equation involves only first order derivatives of $\Psi$, the most general expression one can attribute to the disformal metric making use only of the spinor currents is
\begin{equation}
\widehat{g}^{\mu \nu}=\alpha g^{\mu\nu} + \beta\frac{J^{\mu}J^{\nu}}{J^2} - \gamma\frac{I^{\mu}I^{\nu}}{J^2} + \delta\frac{J^{(\mu}I^{\nu)}}{J^2},
\label{disformalmetric}
\end{equation}
where parentheses indicate symmetrization, $J^2\equiv g^{\mu\nu}J_{\mu}J_{\nu}$ and $\alpha$, $\beta$, $\gamma$ and $\delta$ are arbitrary functions of the scalars $A$ and $B$. Using the relation $\widehat g^{\mu\nu}\, \widehat g_{\nu\lambda} = \delta^\mu_{\ \lambda}$ to guarantee that $\widehat g^{\mu\nu}$ corresponds indeed to a metric tensor, the components of the inverse metric tensor can be written as
\begin{equation}
\label{inv_disformalmetric}
\widehat g_{\mu\nu} = \frac{1}{\alpha} g_{\mu\nu} -\frac{(\alpha+\gamma)\beta+\delta^2}{\alpha \Upsilon^2}\frac{J_{\mu}J_{\nu}}{J^2} + \frac{(\alpha+\beta)\gamma+\delta^2}{\alpha \Upsilon^2}\frac{I_{\mu}I_{\nu}}{J^2} -\frac{\delta}{\Upsilon^2} \frac{J_{(\mu}I_{\nu)}}{J^2},
\end{equation}
with $\Upsilon^2 \equiv (\alpha + \beta)(\alpha + \gamma) + \delta^2$. Metrics in the form (\ref{disformalmetric}) are similar to the class of Finsler geometries introduced by Bekenstein \cite{beken2}, in which the space-time metric depends on some propagating field. We emphasize that $\widehat g^{\mu\nu}$ and $g^{\mu\nu}$ have no gravitational character, that is, the fact that the geometry (\ref{disformalmetric}) and $g^{\mu\nu}$ are not necessarily flat has nothing to do with any geometric theory of gravity.

Let us now introduce the tetrad frames \cite{tetrad1} (or vierbeins \cite{tetrad2}), in order to perform properly the disformal transformation on the spinorial dynamics (see details on the Weyl-Cartan formalism also in \cite{spin_curv}). This simplifies considerably the calculations since we are dealing with an object (spinor) that has a particular internal space and is defined on a curved background, at the same time. We then start by rewriting the quantities defined above in terms of the tetrads. We denote by $\widehat{e}^{\mu}{}_{(A)}$ the tetrad basis acting on the disformal metric $\widehat g^{\mu\nu}$ and $e^{\mu}{}_{(A)}$ the one acting on the target metric $g^{\mu\nu}$. Each orthonormal basis is composed by one time-like and three space-like vectors, with both satisfying the conditions
\begin{equation}
\label{tetrad}
\eta_{AB} = \widehat{g}_{\mu\nu}\,\widehat e^{\mu}{}_{(A)} \, \widehat e^{\nu}{}_{(B)} = g_{\mu\nu}\,e^{\mu}{}_{(A)} \, e^{\nu}{}_{(B)}.
\end{equation}
Note that the Greek indices (running from $0$ to $3$) are lowered and raised by their corresponding space-time metric ($g_{\mu\nu}$ or $\widehat g_{\mu\nu}$) and the Latin labels (running from $1$ to $4$) are lowered and raised with $\eta_{AB}=\mbox{diag}(1,-1,-1,-1)$. Henceforth, the tetrad indexes are denoted without parentheses. Using Eq.\ (\ref{tetrad}), we can define the inverse tetrad bases $e_{\mu}{}^{A}$ and $\widehat{e}_{\mu}{}^{A}$ from the conditions $ e_{\mu}{}^{A} \, e^{\nu}{}_{A} = \widehat e_{\mu}{}^{A} \,\widehat e^{\nu}{}_{A} = \delta_{\mu}^\nu$ and $e_{\mu}{}^{A} \, e^{\mu}{}_{B}= \widehat e_{\mu}{}^{A} \,\widehat e^{\mu}{}_{B}= \delta^{A}_{B}$. The Dirac matrices and the tetrad bases associated with each space-time are such that
\begin{equation}
\gamma^A = \widehat{e}_{\mu}{}^{A} \, \widehat\gamma^\mu= e_{\mu}{}^{A} \, \gamma^\mu,
\label{gammadef}
\end{equation}
where $\gamma^A$'s are the constant Dirac matrices. Then, we can state and prove the following

\newtheorem{lemma}{Lemma}\label{lemma}
\begin{lemma}
Each $\gamma_5$ matrix constructed with $\widehat{\gamma}^{\mu}$, $\gamma^{\mu}$ or $\gamma^A$ corresponds indeed to the same matrix.
\end{lemma}
\begin{proof}
Using Eqs.\ (\ref{tetrad}) and (\ref{gammadef}) to rewrite the determinants of $g_{\mu\nu}$ and $\widehat{g}_{\mu\nu}$ in terms of the tetrad bases in the definition of $\gamma_5$, the proof follows without further problems.
\end{proof}

Furthermore, the $\widehat \gamma^{\mu}$, $\gamma^{\mu}$ and $\gamma^A$ must satisfy their respective closure relation, namely
\begin{equation}
\label{clif_mink}
\{\widehat \gamma^{\mu}, \widehat \gamma^{\nu}\} = 2 \, \widehat g^{\mu\nu}\, {\bf 1},\qquad \{\gamma^{\mu},\gamma^{\nu}\} = 2 \, g^{\mu\nu} {\bf 1}, \quad \mbox{and} \quad \{\gamma^A, \gamma^B\} = 2 \, \eta^{AB}\, {\bf 1},
\end{equation}
where ${\bf 1}$ is the identity element of the algebra and the curly brackets represent the anti-commutation operator.

For consistency with the definitions above, both tetrad bases must be related somehow. Hence, motivated by the algebraic form of the metric (\ref{disformalmetric}), we set
\begin{equation}
\label{mapa_tetr}
\widehat e^{\mu}{}_{A}= a\,e^{\mu}{}_{A} + \Theta^{\mu}{}_{A},
\end{equation}
where the most general expression for $\Theta^{\mu}{}_{A}$ is given by
\begin{equation}
\label{gene_theta}
\Theta^{\mu}{}_{A}= \frac{1}{J^2}[J^{\mu}(bJ_{A}+cI_{A})+I^{\mu}(dJ_{A}+fI_{A})],
\end{equation}
with functions $a, b, c, d$ and $f$ linked to the coefficients of the metric through
\begin{equation}
\label{coef_tet_met}
\alpha=a^2, \quad \beta=b(b+2a)-c^2, \quad \gamma=f(f-2a)-d^2, \quad \mbox{and} \quad \delta=(a-f)c+(a+b)d.
\end{equation}

It is straightforward to verify that the inverse tetrad basis will have a similar form to (\ref{mapa_tetr}), as follows
\begin{equation}
\label{inv_mapa_tetr}
\widehat e_{\mu}{}^{A}= \frac{1}{a} e_{\mu}{}^{A} - \frac{1}{J^2\,\Upsilon} \left[J_{\mu}\left(\frac{b(a-f)+cd}{a}J^{A} + dI^{A}\right) + I_{\mu}\left(c J^{A} + \frac{f(a+b)-cd}{a}I^{A}\right)\right],
\end{equation}
where, in terms of the tetrad coefficients, $\Upsilon\equiv (a+b)(a-f)+cd$. In conclusion, these are the mathematical ingredients on the Weyl-Cartan formalism we need to apply the disformal transformation to the Dirac equation. Nevertheless, we have to analyze first the conservation laws for the spinor currents in both space-times before entering into the details of the dynamical equations for $\Psi$ on the space-times in concern.

\section{Conservation of the currents}\label{cons_curr_sec}
Once we are interested in the disformal invariance of a given dynamics for $\Psi$, it is expected that the action of the disformal transformation on the conservation laws of the spinor currents in a given metric implies the conservation of the corresponding currents in the other metric if these geometries are linked. Though it is not a direct consequence of the map, it can be true when the arbitrary functions of $\widehat{g}^{\mu\nu}$ satisfy some extra constraint equations.

Let us define the Dirac currents $J^{\mu}=\bar{\Psi}\gamma^{\mu}\Psi$ and $\widehat{J}^{\mu}=\bar{\Psi}\widehat{\gamma}^{\mu}\Psi$ associated with the metrics $g_{\mu\nu}$ and $\widehat{g}_{\mu\nu}$, respectively and, using Lemma $1$, the axial currents as $I^{\mu} = \bar{\Psi} \gamma^{\mu} \gamma_5\Psi$ and $\widehat{I}^{\mu} = \bar{\Psi} \widehat{\gamma}^{\mu} \gamma_5\Psi$. This is necessary in order to guarantee the conservation of the currents just as a consequence of the dynamical equation for $\Psi$ in the disformal metric. Thus, imposing the following laws for the currents
\begin{equation}
\label{cons_cur_hat}
\widehat{\nabla}_{\mu}\widehat{J}^{\mu}=0,\quad \mbox{and}\quad \widehat{\nabla}_{\mu}\widehat{I}^{\mu}={\cal S}_d(A,B),
\end{equation}
where $\widehat{\nabla}_{\mu}$ corresponds to the covariant derivative compatible with $\widehat g_{\mu\nu}$ and ${\cal S}_d(A,B)$ is a source term obtained from the combination of the dynamics for $\Psi$ and its complex conjugate both defined in $\widehat g^{\mu\nu}$. If $\Psi$ satisfy the Dirac equation, then ${\cal S}_d$ is proportional to the mass.

Using the Cayley-Hamilton formula, we relate the determinant $\widehat g$ of $\widehat g_{\mu\nu}$ (given by Eq.\ \ref{inv_disformalmetric}) with the determinant $g$ of $g_{\mu\nu}$, that is
\begin{equation}
\widehat g = \frac{g}{a^4\,\Upsilon^2}.
\label{det}
\end{equation}
The definitions of $\widehat J^{\mu}$ and $\widehat I^{\mu}$ together with Eq.\ (\ref{gammadef}) allow us to write $\widehat{J}^{\mu} = (a+b)J^{\mu} + dI^{\mu}$ and $\widehat{I}^{\mu}=(a-f)I^{\mu}-cJ^{\mu}$ and, consequently, Eqs.\ (\ref{cons_cur_hat}) have their analogs in $g_{\mu\nu}$, as follows
\begin{eqnarray}
\widehat{\nabla}_{\mu}\widehat{J}^{\mu}=\frac{a^2 \Upsilon}{\sqrt{-g}}\partial_{\mu}\left\{\frac{\sqrt{-g}}{a^2 \Upsilon}[(a+b)J^{\mu}+dI^{\mu}]\right\}=0,\\
\widehat{\nabla}_{\mu}\widehat{I}^{\mu}=\frac{a^2 \Upsilon}{\sqrt{-g}}\partial_{\mu}\left\{\frac{\sqrt{-g}}{a^2 \Upsilon}[(a-f)I^{\mu}-cJ^{\mu}]\right\}={\cal S}_d,
\end{eqnarray}
where $\partial_{\mu}$ denotes partial derivative. Introducing new variables $X\equiv \ln[(a+b)/a^2 \Upsilon]$, $Y\equiv \ln[(a-f)/a^2 \Upsilon]$, $Z\equiv \ln(d/a^2 \Upsilon)$ and $W\equiv\ln(c/a^2\Upsilon)$ and rearranging the terms, we get
\begin{eqnarray}
&&\nabla_{\mu}J^{\mu} + J^{\mu}\partial_{\mu}X + e^{Z-X}(\nabla_{\mu}I^{\mu} + I^{\mu}\partial_{\mu}Z)=0,\\
&&e^{Y}(\nabla_{\mu}I^{\mu} + I^{\mu}\partial_{\mu}Y) - e^{W}(\nabla_{\mu}J^{\mu} + J^{\mu}\partial_{\mu}W)=\frac{{\cal S}_d}{a^2\Upsilon},
\end{eqnarray}
with $\nabla_{\mu}$ corresponding to the covariant derivative compatible with $g_{\mu\nu}$. In order that the conservation of the currents also holds in the target metric, that is $\nabla_{\mu}J^{\mu}=0$ and $\nabla_{\mu}I^{\mu}={\cal S}_t(A,B)$\footnote{The indexes ``d" and ``t" are used to indicate different sources for the equation of the axial current coming from the dynamics of $\Psi$ in the disformal and target metrics, respectively.}, it is sufficient that the coefficients of the metric $\widehat g_{\mu\nu}$ satisfy the conditions
\begin{subequations}
\label{cond_cons_cur}
\begin{eqnarray}
&&J^{\mu}\partial_{\mu}(e^{X}) + I^{\mu}\partial_{\mu}(e^{Z})=0, \label{cond_cons_cur1}\\[1ex]
&&I^{\mu}\partial_{\mu}(e^{Y}) - J^{\mu}\partial_{\mu}(e^{W}) = \frac{{\cal S}_d+(f-a){\cal S}_t}{a^2\Upsilon}. \label{cond_cons_cur2}
\end{eqnarray}
\end{subequations}
It should be noticed that the existence of conserved currents describing the probability flux of the spinor lead us necessarily to requirements other than the pure disformal map between the dynamics of $\Psi$, contrary to the scalar and electromagnetic cases \cite{erico12,erico13}; otherwise, the map would be mathematically well-defined, but meaningless from the physical point of view. The system (\ref{cond_cons_cur}) involving $a, b, c, d$ and $f$ suggest that the spinor case cannot be naively compared with the previous ones \cite{erico12,erico13}, which impose only algebraic constraints on the metric coefficients. Therefore, disformal transformations applied to spinor fields induce first-order quasi-linear differential equations for the metric coefficient, apart from the Clifford algebra which must be satisfied.


\section{The disformal transformation of the Dirac equation}\label{DB}
In this section, we shall describe how the disformal transformations act on dynamical equations for spinor fields. In particular, we consider a spinor field $\Psi$ satisfying a given dynamics in the disformal metric $\widehat g_{\mu\nu}$ and then applying the disformal map we induce a dynamical equation for $\Psi$ in the target metric $g_{\mu\nu}$, explaining why this metric is called target. After, we also look for conditions on the metric coefficients such that the induced dynamics corresponds to the exact Dirac equation. The way we will apply the disformal transformation here corresponds to a mathematical strategy of course. The dynamics for $\Psi$ is reasonably solvable only if we start from the target metric (which is $\Psi$-independent and given \textit{a priori}) and then construct the disformal metric, where the equation for $\Psi$ defined in terms of $\widehat g_{\mu\nu}$ is automatically verified; otherwise, we need to integrate a nonlinear equation for $\Psi$ in the disformal metric.

From the beginning, we know that the disformal invariance of the massive Dirac equation will be valid only in specific cases, once the disformal transformations contain the conformal ones as a particular case and the latter do not let this equation invariant (see Ref.\ \cite{kastrup} and references therein). As we shall see, the break of the disformal invariance is due solely to the presence of the conformal factor. An alternative to bypass this difficulty would be, for instance, to resort to scenarios where conformal transformations play an important role in the definition of the mass (cf.\ Ref.\ \cite{faraoni}).

The choice of $\widehat g_{\mu\nu}$ and the symmetries associated with $J^{\mu}$ and $I^{\mu}$, evidenced by the equations above, indicate that there is no need of all disformal terms in (\ref{disformalmetric}). So, the disformal transformation can be equally implemented by using only one of the spinor currents. Thus, as a matter of simplicity, we set $\gamma=\delta=0$ in Eqs.\ (\ref{disformalmetric}) and (\ref{inv_disformalmetric}) and define a normalized four-vector $V^{\mu}\equiv J^{\mu}/\sqrt{J^2}$ so that the disformal metric (\ref{disformalmetric}) becomes\footnote{Analogously, one could set $\beta=\delta=0$ and define $V_{\mu}=I_{\mu}/\sqrt{J^2}$, but the special form of (\ref{disformalmetric}) would lead to similar conclusions. Assuming $\beta=\gamma=0$ and $\delta\neq0$ do not provide a genuine disformal transformation once the inverse metric is not given by a binomial expression (see Eq.\ \ref{inv_disformalmetric}).}
\begin{equation}
\label{simp_disf_met}
\widehat g^{\mu\nu}=\alpha g^{\mu\nu}+\beta\, V^{\mu}V^{\nu}.
\end{equation}

In this way, we start with a modified Dirac equation for $\Psi$, where the mass term takes into account the problem introduced by the conformal transformations, written in terms of the metric (\ref{simp_disf_met}) which reads\footnote{Note that when the conformal factor $\alpha$ of Eq.\ (\ref{simp_disf_met}) is constant, we rescue the genuine mass term of the Dirac equation.}
\begin{equation}
\label{Dirac}
i\widehat\gamma^\mu\widehat\nabla_{\mu}\Psi -\sqrt{\alpha}\,m\Psi=0,
\end{equation}
with $\widehat\nabla_{\mu}\equiv\partial_{\mu} - \widehat\Gamma_{\mu}$ and the Fock-Ivanenko connection $\widehat\Gamma_{\mu}$ given by
\begin{equation}
\label{christoffel}
\widehat\Gamma_{\mu} = -\frac{1}{8}\left( [\widehat\gamma^{\alpha},\partial_{\mu}\widehat\gamma_{\alpha}] - \widehat \Gamma^{\rho}_{\alpha\mu}[\widehat\gamma^{\alpha},\widehat\gamma_{\rho}] \right).
\end{equation}
The squared brackets denote the usual commutator and the Christoffel symbol $\widehat \Gamma^{\rho}_{\alpha\mu}$ are constructed with $\widehat g_{\mu\nu}$. We are using units in which $c=\hbar=1$. The aim of this paper can be summarized in the following
\newtheorem{theorem}{Theorem}\label{theorem1}
\begin{theorem}
Let ${\cal M}$ be a space-time with metric tensors $g$ and $\widehat{g}$ and smooth scalar functions $\alpha$ and $\beta$, such that $\alpha$ is positive definite and $\alpha+\beta>0$. Consider a smooth (at least ${\cal C}^2$), normalized, time-like vector field V and a disformal relation between the metrics as
\begin{equation}
\widehat{g}(\ast,\cdot)=\alpha\, g(\ast,\cdot)+\beta\,g(V,\ast)\otimes g(V,\cdot).
\end{equation}
Let $\Psi$ be a Dirac spinor field satisfying the modified Dirac equation\ (\ref{Dirac}) in the disformal metric $\widehat{g}$ and let $V^{\mu} \equiv \bar{\Psi} \gamma^{\mu} \Psi/\sqrt{g^{\alpha\beta} J_{\alpha} J_{\beta}}$ in space-time coordinates, then there exists a class of $\Psi$'s which verifies the Dirac equation (massive or not) written in terms of the target metric $g$.
\end{theorem}

\begin{cor}
If the conformal coefficient $\alpha$ of (\ref{simp_disf_met}) is constant or the spinor is massless, then $\Psi$ is a {\rm disformally invariant solution} of the Dirac operator.
\end{cor}
The complete proof of the theorem and corollary corresponds to the rest of this and the next sections.

Let us introduce a tetrad basis with the same disformal symmetry as its associate metric (\ref{simp_disf_met}), that is
\begin{equation}
\label{mapa_tetr_simpl}
\widehat e^{\mu}{}_{A}= a\,e^{\mu}{}_{A} + bV^{\mu}V_{A},\quad \mbox{and} \quad \widehat e_{\mu}{}^{A}= \frac{1}{a} e_{\mu}{}^{A} - \frac{b}{a(a+b)} V_{\mu}V^{A}.
\end{equation}
The Fock-Ivanenko connection (\ref{christoffel}) can be rewritten as
\begin{equation}
\widehat\Gamma_{\mu} = \widehat e_\mu{}^A \widehat\Gamma_{A}, \qquad \mbox{with}\qquad \widehat\Gamma_{A}=-\frac{1}{8}\,\widehat\gamma_{BCA}[\gamma^{B},\gamma^{C}],
\end{equation}
where $\widehat\gamma_{ABC}$ is the spin connection defined as
\begin{equation}
\label{C_ABC}
\widehat\gamma_{ABC} = \frac{1}{2} (\widehat C_{ABC} - \widehat C_{BAC} - \widehat C_{CAB}) \quad \mbox{and} \quad
\widehat C_{ABC} = \widehat e_{\nu A} (\widehat{e}^{\mu}{}_{C}\,\partial_{\mu}\widehat{e}^{\nu}{}_{B} - \widehat{e}^{\mu}{}_{B}\,\partial_{\mu}\widehat{e}^{\nu}{}_{C}),
\end{equation}
with symmetry $\widehat C_{ABC} = - \widehat C_{ACB}$, implying $\widehat\gamma _{ABC} = - \widehat\gamma_{BAC}$.

The next step is to use Eq.\ (\ref{mapa_tetr_simpl}) and rewrite Eq.\ (\ref{Dirac}) completely in terms of the objects defined at the target metric. Thus, a first manipulation of Eq.\ (\ref{Dirac}) yields
\begin{equation}
\label{tetr_Dirac}
i\gamma^A(\widehat\partial_{A} - \widehat\Gamma_{A})\Psi \ -am\Psi = \ i\gamma^A\left(a\partial_A+bV_AV^B\partial_B +\frac{1}{8}\widehat\gamma_{BCA}[\gamma^{B},\gamma^{C}]\right)\Psi -am\Psi=0,
\end{equation}
where $\widehat \partial_A\equiv\widehat e^\mu{}_{A} \partial_{\mu}= a\partial_A +b V_AV^B\partial_B$. Then, using the algebraic identity between the $\gamma^A$ matrices
$$\gamma^A\gamma^B\gamma^C = \eta^{AB}\gamma^C + \eta^{BC}\gamma^A - \eta^{AC}\gamma^B + i\epsilon^{ABC}{}_{D}\gamma^D\gamma_5$$
to compute the expression $\widehat\gamma_{BCA}\gamma^A\gamma^B\gamma^C$, we get
\begin{equation}
i\left[\gamma^{A} \partial_{A} - m +\frac{b}{a}\gamma^A V_A V^B \partial_{B} +\frac{1}{4a}\left(2\widehat C^{A}{}_{BA}\gamma^B +\frac{i}{2} \widehat C_{ABC}\epsilon^{ABC}{}_{D}\gamma^D\gamma_5\right)\right] \, \Psi =0.
\end{equation}
From Eq.\ (\ref{C_ABC}), we encounter that
$$\widehat C^{A}{}_{BA}=aC^{A}{}_{BA}-\partial_{\mu}a\left[\frac{3a+2b}{a+b}\,e^{\mu}{}_B +\frac{b\,(4a+3b)}{a\,(a+b)}\,V^{\mu}V_B \right]+\nabla_{\mu}(bV^{\mu})V_B + b\dot V_B$$
and
$$\widehat C_{ABC}\epsilon^{ABC}{}_{D}=\left[a C_{ABC}+2b\, e_{\nu A}V_C\left( \dot e^{\nu}{}_B + \frac{b}{a + b}\nabla_{\mu}V^{\nu}e^{\mu}{}_B\right) \right]\epsilon^{ABC}{}_{D},$$
where dot ($\,\dot\,\,$) means covariant derivative in the target metric projected along $V^{\mu}$. 

Finally, using the condition coming from the conservation of the vector current (\ref{cond_cons_cur1}) restricted to the case we are dealing with, which is merely
\begin{equation}
\label{cond_corr}
\dot{a}=0,
\end{equation}
the equation (\ref{Dirac}) written in the disformal geometry becomes the following in the target space-time
\begin{equation}
\label{disf_dirac}
\begin{array}{l}
i\gamma^{A} \nabla_{A}\Psi -m\Psi +\fracc{ib}{a}\gamma^{A} V_{A}V^{B} \nabla_{B}\Psi - \fracc{ib}{2a}\left[\fracc{3a+2b}{b(a+b)}\,\partial_Ba -\fracc{\dot{b}}{b}V_B -(\nabla_{\mu}V^{\mu})V_B -\dot{V}_B\right]\gamma^B\Psi+\\[2ex]
+\fracc{b}{4a}\left\{e_{\nu}{}^{B}V_D \left[\dot e^{\nu}{}_{C} + \fracc{b}{a+b} \nabla_{\mu}V^{\nu}\, e^{\mu}{}_{C} \right] \epsilon_{AB}{}^{CD} \gamma^{A} \gamma_5 \right\}\Psi =0.
\end{array}
\end{equation}
Note that this is a nonlinear dynamical equation for $\Psi$ on ${\cal M}$ endowed with $g_{\mu\nu}$: its first two terms correspond to the Dirac operator $i\gamma^A\nabla_A-m{\bf 1}$ fully defined on the target metric and all the other self-interacting terms are originated by the relation between the tetrad bases. as far as we know, Eq.\ (\ref{disf_dirac}) does not fit any well-known nonlinear dynamics for a Dirac spinor \cite{fushchych}. Notwithstanding, we have shown that all solutions of Eq.\ (\ref{Dirac}) defined in the disformal metric given by (\ref{simp_disf_met}) are the same for this highly nonlinear equation in the target metric. In the next section, we will see that there are special classes of solutions for this equation which satisfies the exact Dirac equation in the target metric.

\section{Classes of Inomata-type solutions}\label{spec-clas}
With no loss of generality, we consider an arbitrary point $p\in{\cal M}$ and Riemannian normal coordinates around $p$, namely, the target metric is the Minkowski one at $p$. Consequently, the matrix representation of the tetrad basis $e^{\mu}{}_A$ is the Kronecker delta and the covariant derivatives reduce to partial derivatives ($\nabla_\mu\rightarrow\partial_\mu$). In this case, Eq.\ (\ref{disf_dirac}) is considerably simplified, yielding
\begin{equation}
\label{simp_disf_dirac}
i\gamma^{A} \partial_{A}\Psi -m\Psi +\fracc{ib}{a}\gamma^{A} V_{A}V^{B} \partial_{B}\Psi-\fracc{ib}{2a}\left[\frac{3a+2b}{b(a+b)}\,\partial_Ba-\frac{\dot{b}}{b}V_B-(\partial_{\mu}V^{\mu})V_B-\dot{V}_B\right]\gamma^B\Psi-\fracc{b^2}{2a(a+b)}\omega_A\gamma^{A} \gamma_5\Psi =0,
\end{equation}
where $\omega^A\doteq-\frac{1}{2}\epsilon^{ABCD}\omega_{BC}V_{D}$ is the vorticity vector and $\omega_{BC}\doteq\frac{1}{2}h_{B}{}^{\mu}h_{C}{}^{\nu}(\partial_{\nu}V_{\mu}-\partial_{\mu}V_{\nu})$ is the vorticity tensor both associated with $V_A$.

In Ref.\ \cite{Inomata}, Inomata found classes of $\Psi$s satisfying the Heisenberg equation by assuming that the derivative of the spinor field could be written as a linear combination of the elements of the Clifford algebra using semilinear coefficients depending on $\Psi$. However, in our case, some self-interacting terms depend also on derivatives of $\Psi$, in particular, the vorticity associated with the Dirac current and, thus, Inomata's
condition cannot be applied in general. Therefore, we propose a generalization of it assuming as an {\it Ansatz} that the coefficients of the linear combination of the elements of the algebra could involve also quasi-linear terms, that is

\begin{equation}
\label{der_cond}
\partial_B\Psi = \left(s_0 - \frac{mB\,\gamma_5}{3\sqrt{J^2}}\right)V_B\Psi - \left(s_0+\frac{imA}{\sqrt{J^2}}-\frac{4mB\,\gamma_5}{3\sqrt{J^2}}\right)\frac{(A-iB\gamma_5)\gamma_B}{4\sqrt{J^2}}\Psi +\gamma_B\gamma_C\omega^C(s_1+s_2\gamma_5)\Psi + \omega_B(s_3 + s_4\gamma_5)\Psi,
\end{equation}
where $s_0$ for $j=0,...,4$ are arbitrary functions of $\Psi$ to be determined. If we apply $\gamma^B$ to Eq.\ (\ref{der_cond}), we get the Dirac operator on the left-hand side. In order to have the Dirac equation satisfied, that is $i\gamma^{A} \partial_{A}\Psi-m\Psi=0$, we impose the following constraints on the free functions
$$s_3=-4s_1,\quad \mbox{and} \quad s_4=-4s_2.$$
Therefore, $\Psi$ given by Eq.\ (\ref{der_cond}) is a solution of the massive Dirac equation in the vicinity of $p\in{\cal M}$.

However, we need to guarantee that this class also verify the remaining equation constituted by the other terms of\ (\ref{simp_disf_dirac}). Then, applying $\gamma^{A}V_AV^B$ to Eq.\ (\ref{der_cond}) and substituting the outcome into Eq.\ (\ref{simp_disf_dirac}), we get
$$\left[\left(\fracc{3s_0}{4}-\fracc{imA}{4\sqrt{J^2}}+\fracc{\dot{b}}{2b} + \frac{1}{2}\partial_{\mu}V^{\mu}\right) V_A + s_1\omega_A  - \frac{3a+2b}{b(a+b)} \, h_{A}{}^{\hspace{-0.1cm}{}^{C}}\partial_{{}_{C}} \,a + \frac{1}{2}\dot V_A\right]\gamma^A\Psi +\left(is_2 -\frac{b}{2(a+b)}\right)\omega_A\gamma^A\gamma_5\Psi=0,$$
where we gather the terms according to the linear independence of the algebra elements and the orthogonality with respect to $V_A$. At the end, we obtain three equations determining completely the remaining free functions of (\ref{der_cond})
\begin{subequations}
\label{cond_mink_disf}
\begin{eqnarray}
s_0&=&\fracc{imA}{3\sqrt{J^2}}+\fracc{2}{3}\left(\ln \fracc{\sqrt{J^2}}{b}\right)^{\bullet},\label{cond_mink_disf1}\\[1ex]
s_2&=& -\fracc{ib}{2(a+b)},\label{cond_mink_disf2}\\[1ex]
\dot V_{{}_{A}}&=&\frac{3a+2b}{b(a+b)}\,h_{{}_{A}}{}^{\hspace{-0.15cm}{}^{C}}\partial_{{}_{C}}\,a - 2s_1\omega_{{}_{A}}\label{cond_mink_disf3}
\end{eqnarray}
\end{subequations}
plus the conservation law associated with the axial current (\ref{cond_cons_cur2}), with ${\cal S}_d=2amB$ and ${\cal S}_t=2mB$, restricted to the case (\ref{mapa_tetr_simpl}), that is
\begin{equation}
\label{cond_axial_mink}
2\frac{a'}{a}+\frac{a'+b'}{a+b}=0,
\end{equation}
where prime means covariant derivative projected along to $I^{\mu}$. Note that Eqs.\ (\ref{cond_mink_disf1}) and (\ref{cond_mink_disf2})  provides $s_0$ and $s_2$ algebraically in terms of $\Psi$ and the tetrad coefficients, but instead we have to solve the differential equation (\ref{cond_mink_disf2}) to find $s_1$. Due to Helmholtz's decomposition (also known as the fundamental theorem of the vector calculus) and the relation between $a$ and $b$ given by Eq.\ (\ref{cond_axial_mink}), this equation always admits a solution for $s_1$, once on the right hand side the first term is the gradient of a function of $a$ and the second one is a rotational. We have also used the conservation law for $J^{\mu}$ to write $\partial_{\mu}J^{\mu}=0$ as $\partial_{\mu}V^{\mu}=-(\ln \sqrt{J^2})^{\bullet}$.

It should be remarked that the disformal invariance of the Dirac equation can be restored without regarding to Riemannian normal coordinates if we modify the special class of solutions (\ref{der_cond}) properly through
\begin{equation*}
\partial_B\longrightarrow\nabla_B, \quad \mbox{and} \quad \omega_A\longrightarrow\omega_A+\epsilon_{AB}{}^{CD}e_{\nu}{}^{B}\,\dot e^{\nu}{}_{C}\,V_D.
\end{equation*}
Since the set of vectors $\{e^{\mu}{}_{B}\}$ form an orthonormal basis, $\dot e^{\mu}{}_{C}$ is precisely the acceleration associated with the tetrad frame and, therefore, the modifications caused by this in the equations above are not purely mathematical, but instead they have a clear physical interpretation in terms of the kinematical quantities of a congruence of curves.

Summarizing the results obtained above, we have shown that the disformal transformation of a modified Dirac equation written in $\widehat{g}_{\mu\nu}$ leads to a nonlinear equation for $\Psi$ in $g_{\mu\nu}$. It means that the solutions of the former also satisfy the latter. Then, we select a sub-class of $\Psi$s satisfying a generalized Inomata's condition in order that this nonlinear equation in $g_{\mu\nu}$ reduces to the Dirac equation
\begin{equation}
\label{target-dirac}
i\gamma^{\mu} \nabla_{\mu}\Psi-m\Psi=0.
\end{equation}
Therefore, we have provided sufficient conditions for the disformal map of the Dirac equation, as claimed by Theorem $1$, achieving besides the disformal invariance in case Corollary $1$ holds true. It is important to note that the search for explicit solutions of the Dirac equation satisfying the generalized Inomata condition (\ref{der_cond}) is a hard task. In particular, this quasi-linear system of PDEs is not verified for simple solutions of the Dirac equation.

\section{The case of light-like Dirac current}\label{null-case}
The results present above can be easily extended for the case when the Dirac current is light-like ($A,B=0$). However, the expression (\ref{simp_disf_met}) for the disformal metric is not appropriate because $V^{\mu}$ has $J^2$ in the denominator. Therefore, the procedure corresponds to replace $\beta\longrightarrow\beta J^2$ and $b\longrightarrow bJ^2$ in the disformal metric and its corresponding tetrad basis and substitute $V^{\mu}$ explicitly in terms of the Dirac current $J^{\mu}$. We then have
\begin{equation}
\label{simp_disf_met_null}
\widehat g^{\mu\nu}=\alpha g^{\mu\nu}+\beta\, J^{\mu}J^{\nu},\quad \mbox{and}\quad \widehat g_{\mu\nu}=\frac{1}{\alpha} g_{\mu\nu}-\fracc{\beta}{\alpha^2}\, J_{\mu}J_{\nu},
\end{equation}
with tetrad basis
\begin{equation}
\label{mapa_tetr_simpl_null}
\widehat e^{\mu}{}_{A}= a\,e^{\mu}{}_{A} + bJ^{\mu}J_{A},\quad \mbox{and} \quad \widehat e_{\mu}{}^{A}= \frac{1}{a} e_{\mu}{}^{A} - \frac{b}{a^2} J_{\mu}J^{A}.
\end{equation}
Once light-like currents are linked only to massless particles, for instance neutrinos, the disformal map consists in rewriting the massless Dirac equation defined in the disformal metric (Eq.\ \ref{Dirac} with $m=0$) in terms of the objects associated with the target metric. A straightforward calculation yields that the expression corresponding to Eq.\ (\ref{disf_dirac}) is
\begin{equation}
\label{disf_dirac_null}
i\gamma^{A} \nabla_{A}\Psi -\fracc{i}{2}(\dot J_B+3\partial_B a)\gamma^B\Psi -\fracc{b}{4a}\sigma_A\gamma^A\gamma_5\Psi=0,
\end{equation}
where $\sigma_A\equiv\epsilon_{AB}{}^{CD}e_{\nu}{}^{B}J_D\dot e^{\nu}{}_{C}$.

Thus, the sufficient conditions to have the disformal invariance of the Dirac equation are
\begin{equation}
\label{cons_null}
\partial_B(\ln a)=-\frac{1}{3b}\dot J_{B}, \quad \mbox{and} \quad \sigma_A=0.
\end{equation}
Note that these equations take into account only space-time objets, i.e., no mention to the elements of the Clifford algebra, constraining only one of the two coefficients of the disformal metric and the tetrad basis of the target metric.

There is a simple example in this case, if we assume the target metric as the Minkowski space-time in Cartesian coordinates (i.e., $e^{\nu}{}_{B}$ is the Kronecker delta) and $\Psi$ as a plane wave solution of the massless Dirac equation in this background:
\begin{equation}
\label{neut_spin_null}
\Psi=\Psi_0\left(\begin{array}{c}
1\\0\\-1\\0
\end{array}\right)e^{-iE(t+z)},
\end{equation}
where $\Psi_0$ is a constant amplitude and $E$ is the energy of the particle described by $\Psi$. The vector and axial currents are $J_{B} = 2|\Psi_0|^2 (1,0,0,1)= -I_{B}$. Therefore, Eqs.\ (\ref{cons_null}) are identically satisfied if we set $a$ equal to a constant. The conservation law of the axial current (\ref{cond_axial_mink}) impose that $b=b(t-z,x,y)$.

Finally, introducing auxiliary coordinates $\eta=(t-z)/a$, $\xi=(t+z)/a$, $\tilde x=x/a$ and $\tilde y=y/a$, the disformal metric associated with (\ref{neut_spin_null}) gives the infinitesimal line element
\begin{equation}
\label{disf_null}
\widehat{ds^2}=d\eta d\xi-d\tilde x^2-d\tilde y^2-\tilde\beta(\eta,x,y)d\xi^2,
\end{equation}
with $\tilde\beta=4\beta|\Psi_0|^4/a^2$. The only non-zero component of the Riemann curvature is $R_{\eta\xi\eta\xi}=\tilde \beta_{,\eta\,\eta}/2$, showing that the disformal metric is not flat in general. If we write
$$\partial_{B}\Psi=-iE\,\frac{J_B}{J_0}\Psi,\quad \mbox{and} \quad \gamma^{A}\widehat\Gamma_{A}\Psi=-\frac{1}{2} b_{,B} J^{B} J_A\gamma^A \Psi,$$
and use Pauli-Kofink (\ref{pkofink}) which gives $\gamma_AJ^A\Psi=0$, it is easy to see that Eq.\ (\ref{neut_spin_null}) satisfies the Dirac equation in the disformal metric.

\section{On the disformal group structure}\label{disf_group}
Now let us analyze one of the mathematical structures behind the disformal transformations. In the same manner as obtained for the scalar \cite{erico12} and the electromagnetic field \cite{erico13}, the disformal metrics associated with the spinor fields can be seen as members to a two-parameter group structure and the results of this section are valid for any disformal transformation of the form (\ref{simp_disf_met}). Suppose, to begin with, that we have fixed a space-time (smooth manifold $\cal{M}$ and a background Lorentzian metric $g$) and a spinor field $\Psi$ with its corresponding Dirac current $J^{\mu}$. Thus, we can make the identification
\begin{gather}
\label{groupeq}
\widehat{g}^{\mu\nu} = \alpha g^{\mu\nu} + \fracc{\beta}{g^{\gamma\delta}J_{\gamma}J_{\delta}} g^{\mu\rho}g^{\nu\sigma}J_{\rho}J_{\sigma} \longleftrightarrow \lceil \alpha,\beta\rceil g,	
\end{gather}
and use this to define the action of $\lceil \alpha,\beta\rceil$ on the metric $g$.
We shall denote  the set of all such $\lceil \alpha,\beta\rceil$ by $\mathfrak{G}$. That is,
\begin{eqnarray}
\mathfrak{G} = \{ \lceil \alpha,\beta\rceil\,|\,\ \alpha >0\,\, \mbox{and}\,\, \alpha+\beta>0\}.	
\end{eqnarray}

The set $\mathfrak{G}$ with the operation $\star$ defined by
\begin{gather}
(\lceil \alpha,\beta\rceil \star \lceil \alpha',\beta'\rceil)g \doteq\lceil \alpha\alpha',\alpha'\beta+\beta'\beta+\alpha\beta'\rceil g
\end{gather}
is a group. Before proving this statement, let us investigate the meaning of this composition law: when we first evaluate $\lceil \alpha',\beta'\rceil g$, we obtain Eq.\ (\ref{groupeq}) with primed $\alpha$ and $\beta$, i.e., the contra-variant components of a new metric tensor $\widehat{g}$. Then, the composition law is defined in such a way that
\begin{gather}
(\lceil \alpha,\beta\rceil \star \lceil \alpha',\beta'\rceil)g = \lceil \alpha,\beta\rceil \widehat{g}
\end{gather}
is true.
The proof that $(\mathfrak{G},\star)$ is a group is straightforward and it is shown bellow:
\begin{enumerate}
\item{Existence of the identity element: there exists an element   $\lceil 1,0\rceil \in\mathfrak{G}$ such that  for all $\lceil\alpha,\beta\rceil\in\mathfrak{G}$ holds
\begin{gather}
	\lceil 1,0\rceil\star \lceil\alpha,\beta\rceil=\lceil\alpha,\beta\rceil\star \lceil 1,0\rceil=\lceil\alpha,\beta\rceil.
\end{gather}
}
\item{Existence of the inverse element: for each $\lceil\alpha,\beta\rceil \in \mathfrak{G}$ there exists another element  $\left\lceil \frac{1}{\alpha},-\frac{\beta}{\alpha(\alpha+\beta)} \right\rceil \in\mathfrak{G}$ such that
\begin{gather}
\left\lceil \frac{1}{\alpha},-\frac{\beta}{\alpha(\alpha+\beta)} \right\rceil\star\,\lceil\alpha,\beta\rceil= \lceil\alpha,\beta\rceil \star \left\lceil \frac{1}{\alpha},-\frac{\beta}{\alpha(\alpha+\beta)} \right\rceil  =\lceil 1,0\rceil.
\end{gather}
}
\item{Associativity: if we consider $\lceil\alpha,\beta\rceil$, $\left\lceil \alpha', \beta' \right\rceil$,$\left\lceil \alpha'', \beta'' \right\rceil \in \mathfrak{G}$ , we have
\begin{eqnarray}
\label{assoc}
\left\lceil \alpha'', \beta'' \right\rceil\star\, \left(\left\lceil \alpha', \beta' \right\rceil\star\, \lceil\alpha,\beta\rceil\right)&=&\left\lceil \alpha'', \beta'' \right\rceil\star\,\left\lceil \alpha'\alpha,\alpha'\beta+\beta'\beta+\alpha\beta' \right\rceil\nonumber\\
&=&\left\lceil \alpha''\alpha'\alpha,\alpha''(\alpha'\beta+\beta'\beta+\alpha\beta')+\beta''(\alpha'\beta+\beta'\beta+\alpha\beta')+\beta''\alpha\alpha' \right\rceil\nonumber\\[2ex]
&=&\left\lceil \alpha''\alpha',\alpha''\beta'+\beta''\beta'+\beta''\alpha'\right\rceil\star \lceil\alpha,\beta\rceil\nonumber\\
&=&\left(\left\lceil \alpha'', \beta'' \right\rceil\star \left\lceil \alpha', \beta' \right\rceil\right)\star \lceil\alpha,\beta\rceil.
\end{eqnarray}}
\end{enumerate}
Besides, it is direct to check that the disformal group $(\mathfrak{G},\star)$ is in fact an Abelian group.

On the other hand, using Eq.\ (\ref{inv_disformalmetric}), we find that the inverse of $\lceil \alpha, \beta \rceil g$ is written as
\begin{gather}
(\lceil \alpha, \beta \rceil g)^{-1} = \frac{1}{\alpha}g_{\mu\nu} - \frac{\beta}{\alpha(\alpha + \beta)}V_{\mu}V_{\nu},
\end{gather}
which motivates the definition of another operation $\lfloor \alpha, \beta \rfloor$ acting on $g$ by
\begin{gather}
\lfloor \alpha, \beta \rfloor g = \frac{1}{\alpha}g_{\mu\nu} - \frac{\beta}{\alpha(\alpha + \beta)}\frac{J_{\mu}J_{\nu}}{g^{\gamma\delta}J_{\gamma}J_{\delta}},
\end{gather}
and an analogous composition law $\odot$ given by
\begin{gather}\label{floorop}
(\lfloor \alpha,\beta\rfloor \odot \lfloor \alpha',\beta'\rfloor)g \doteq\lfloor \alpha\alpha',\alpha'\beta+\beta'\beta+\alpha\beta'\rfloor g.
\end{gather}
It should be clear in the notation that $\lceil \alpha, \beta \rceil g$ and $\lfloor \alpha, \beta \rfloor g$ are related to the contra-variant and the covariant form of the metric $g$, respectively. As the reader should note, the set  $\mathfrak{H}=\{ \lfloor \alpha,\beta\rfloor\,|\,\ \alpha >0\,\, \mbox{and}\,\, \alpha+\beta>0\}$ with the operation $\odot$ acting on $g$ according to (\ref{floorop}) is also an Abelian group and  $((\lceil \alpha, \beta \rceil \star \lceil\alpha',\beta'\rceil) g)^{-1} = (\lfloor\alpha,\beta\rfloor\odot\lfloor\alpha',\beta'\rfloor)g$. In other words, the inverse of the composition of two disformal transformations is the same as the composition of the inverses of those disformal transormations. This is a crucial fact that allow us, when dealing with disformal transformations, to define with no ambiguity  $\widehat{g}_{\mu\nu}$ from the given $\widehat{g}^{\mu\nu}$. Furthermore, the application
\begin{eqnarray*}
\phi : (\mathfrak{G},\star)\longrightarrow(\mathfrak{H},\odot)\\
\lceil\alpha,\beta\rceil\mapsto\lfloor\alpha,\beta\rfloor,
\end{eqnarray*}
satisfies
\begin{gather}
\phi(\lceil\alpha,\beta\rceil\star\lceil\alpha',\beta'\rceil) = \phi(\lceil\alpha,\beta\rceil)\odot\phi(\lceil\alpha',\beta'\rceil),
\end{gather}
hence it is a group homomorphism. In fact, $\phi$ is a group isomorphism, as expected.

As particular examples of disformal sub-groups, we have the case when all conformal coefficients are equal to $1$ which makes the disformal metrics similar to those from the spin-2 field theory formulation, but with finite inverse metric, and the cases in which the disformal coefficients are zero ($\beta's=0$), coinciding with the usual conformal group.

\section{Concluding remarks}\label{conclusion}
We have analyzed the action of the disformal transformations on the case of propagating spinor fields making use of the Weyl-Cartan formalism. In particular, we have shown that generalizing the Inomata condition it is possible to find a class of solutions which let the Dirac equation almost invariant under these maps, up to a conformal factor in the mass term. That is, they verify this equation in the disformal and target metrics once a set of conditions are fulfilled. However, explicit expressions for $\Psi$ satisfying the hypotheses of Theorem $1$ are not simple. Preliminary attempts have pointed out that if $\Psi$ is a superposition of plane waves this could be solved, but we postpone the complete analysis for the near future.

The situation in which the norm of $J^{\mu}$ is zero could be easily obtained from previous results by modifying the disformal term in the metric and the tetrad basis, and then taking the limit $J^{2} \rightarrow 0$. In comparison with the case $J^2\neq0$, this leads to a simpler set of equations for $\Psi$ in the target metric and requiring only two conditions upon space-time objects for the disformal invariance of the massless Dirac equation. In this context, a simple example in terms of plane waves allows for the construction of a curved disformal metric. It should be noticed that Eq.\ (\ref{simp_disf_met_null}) correspond to the Kerr-Schild metrics of general relativity, which means that the disformal invariance of the Dirac equation could be useful in the study of propagating spinor fields on these backgrounds.

At the end, we have demonstrated that the disformal transformations in the spinorial case satisfy an Abelian group structure. Therefore, as the conformal transformations do, the disformal transformations appear as a new symmetry that all fundamental field (scalar, vector and now spinor) satisfying a PDE can be invariant under certain circumstances. To be promoted as a fundamental symmetry of nature, the complete physical meaning of this kind of transformation is still under investigation.

\section*{Acknowledgments}
We acknowledge the participants of the \textit{Seminario Informale} at Sapienza University for their comments, in particular, we wish to thank Dr.\ Grasiele Santos for suggesting improvements to the manuscript. The authors are supported by the CAPES-ICRANet program (BEX 13956/13-2, 14632/13-6, 15114/13-9).


\begin{thebibliography}{100}
\bibitem{beken1}
J.D. Bekenstein, in {\em Proceedings of the Sixth Marcel Grossmann Meeting on General Relativity}, eds.\ H.\ Sato and T. Nakamura, World Publishing,
Singapore, (1992);
\bibitem{beken2}
J.D. Bekenstein, {\em Phys.\ Rev.\ D} {\bf 48} 3641 (1993). 
\bibitem{beken_mond}
J.D. Bekenstein, {\em Phys.\ Rev.\ D} {\bf 70} 083509 (2004), [Erratum-ibid. D {\bf 71} 069901 (2005)].
\bibitem{milgrom}
M. Milgrom, {\em Phys.\ Rev.\ D} {\bf 80} 123536 (2009).
\bibitem{amelino}
G. Amelino-Camelia, {\em Liv.\ Rev.\ Rel.} {\bf 16} 5 (2013).
\bibitem{clifton}
T. Clifton, P.G. Ferreira, A. Padilla and C. Skordis, {\em Phys.\ Rep.} {\bf 513} 1 (2012).
\bibitem{dario}
D. Bettoni, S. Liberati, {\em Phys.\ Rev.\ D} {\bf 88} 084020 (2013).
\bibitem{nemanja}
N. Kaloper, {\em Phys.\ Lett.\ B} {\bf 583} 1 (2004).
\bibitem{bitt_nov_faci}
E. Bittencourt S. Faci and M. Novello, {\em Int.\ J.\ Mod.\ Phys.\ A} {\bf 29} 1450145 (2014);
\bibitem{nov_bit}
M. Novello and E. Bittencourt, {\em Int.\ J.\ Mod.\ Phys.\ A} {\bf 29} 1450075 (2014);
\bibitem{nov_bit_gordon}
M. Novello and E. Bittencourt, {\em Phys.\ Rev.\ D} {\bf 86} 124024 (2012).
\bibitem{nov_bit_drag}
M. Novello and E. Bittencourt, {\em Gen.\ Rel.\ Grav.} {\bf 45} 1005 (2013).
\bibitem{novetgoul}
M. Novello and E. Goulart, {\em Class.\ Quantum\ Grav.} {\bf 28} 145022 (2011).
\bibitem{goulart}
E. Goulart, M. Novello, F.T. Falciano and J.D. Toniato, {\em Class.\ Quantum\ Grav.} {\bf 28} 245008 (2011).
\bibitem{erico12}
F.T. Falciano and E. Goulart, {\em Class.\ Quantum Grav.} {\bf 29} 085011 (2012). 
\bibitem{erico13}
E. Goulart and F.T. Falciano, {\em Class.\ Quantum Grav.} {\bf 30} 155020 (2013). 
\bibitem{Inomata}
A. Inomata and W.A. McKinley, {\em Phys.\ Rev.} {\bf140} 1467 (1965); A. Inomata, {\em Phys.\ Rev.\ D} {\bf 18} 3552 (1978).
\bibitem{spin_curv}
I.D. Soares, {\em Proc.\ II Brazilian School of Cosmology and Gravitation}, Ed. M. Novello, impressed by J. Sasson \& Cia. Ltd., Rio de Janeiro (1980).
\bibitem{feynman}
R. P. Feynman, F. B. Morinigo e W. G. Wagner, \textit{Feynman\ lectures on gravitation}, Addison Wesley Pub. Company, Massachusetts, (1995).
\bibitem{tetrad1}
E. Cartan {\em Ann.\ Sci.\ Ec.\ Norm.\ Sup.} {\bf 40} 325 (1923);
\bibitem{tetrad2}
H. Weyl {\em Zeit.\ Phys.} {\bf 56} 330 (1929).
\bibitem{kastrup}
H. A. Kastrup, {\em Annalen\ Phys.} {\bf 17} 631 (2008).
\bibitem{faraoni}
V. Faraoni and S. Nadeau, {\em Phys.\ Rev.\ D} {\bf 86} 023501 (2007).
\bibitem{fushchych}
W.I. Fushchych, R.Z. Zhdanov, {\em Phys.\ Rep.} {\bf 172} 4 p.123 (1989).
\end{thebibliography}
\end{document}